\documentclass[journal]{IEEEtran}
\usepackage{graphics}
\usepackage{amssymb}
\usepackage{color}
\usepackage{xcolor}
\usepackage{graphicx}
\usepackage{amsthm}
\usepackage{paralist}
\usepackage{verbatim}
\usepackage{chemarr}
\usepackage{braket}
\usepackage[normalem]{ulem}
\newtheorem{theorem}{Theorem}
\newtheorem{conjecture}{Conjecture}
\newtheorem{lemma}[theorem]{Lemma}

\def\EQ#1{\begin{equation}\begin{aligned}#1\end{aligned}\end{equation}}

\newcommand{\ketbra}[2]{|#1\rangle\langle #2|}

\begin{document}

\title{Quantum Rate-Distortion Coding of Relevant Information}
\author{Sina~Salek,
        Daniela~Cadamuro,
        Philipp~Kammerlander
        and~Karoline~Wiesner
\thanks{Sina Salek is with the Department of Computer Science, University of Oxford, Wolfson Building, Parks Road, UK. Email: salek.sina@gmail.com. He was with the University of Hong Kong, Hong Kong SAR, and the University of Bristol, UK, for part of this research.}
\thanks{Daniela Cadamuro is with Zentrum Mathematik, Technische Universit\"at M\"unchen, 85748 Garching,
Germany. She was with the University of Bristol for part of this research.}
\thanks{Philipp Kammerlander is with Institut f\"ur Theoretische Physik, ETH Z\"urich, 8093 Z\"urich, Switzerland.}
\thanks{Karoline Wiesner is with School of Mathematics, University of Bristol, Bristol BS8 1TW, United Kingdom.}}

%
%
%
%
%
\maketitle
\begin{abstract}
Rate-distortion theory provides bounds for compressing data produced by an information source to a specified encoding rate that is strictly less than the source's entropy. This necessarily entails some loss, 
or distortion, between the original source data and the best approximation after decompression.
The so-called Information Bottleneck Method is designed to compress only `relevant' information. Which information is relevant is determined by the correlation between the data being compressed and a variable of
interest, so-called side information.
 In this paper, an Information Bottleneck Method is introduced for the compression of quantum data. The channel communication picture is used for compression and decompression. 
The rate of compression is derived using an entanglement assisted protocol with classical communication, and under an unproved conjecture that the rate function is convex in the distortion parameter. The optimum channel achieving this rate for a 
given input state is characterised. The conceptual difficulties arising due to differences in the mathematical formalism between quantum and classical probability theory are discussed and solutions are presented.
\end{abstract}

\section{Introduction \label{intro}}

One of the most central results in classical information theory is Shannon's data compression theorem \cite{Shannon1948} which gives a fundamental limit on lossless compressibility of data. Due to statistical redundancies, data can be compressed at a rate bounded below by the source entropy, such that after decompression the full information is recovered without loss.  
Rate-distortion theory (RDT) is the branch of information theory that  compresses the data produced by an information source down to a specified encoding rate that is strictly less than the source entropy  \cite{Shannon1959}. This necessarily entails some loss, or distortion, between the original source data and the best approximation after decompression, according to some distortion measure. RDT is frequently used in multimedia data compression where a large amount of data can be discarded without any noticeable change to a listener or viewer.

Whilst RDT is an important and widely used concept in information theory, there are cases where only part of the information in the data to be compressed is relevant. For instance in speech signal processing one might be interested only in information about the spoken words in audio data. 
The Information Bottleneck Method (IBM), introduced by Tishby \emph{et al.}, achieves a lossy compression rate even lower than the rate given by RDT by keeping only `relevant' information \cite{Tishby1999}. Which information is relevant is determined by the correlation between the data being compressed and a variable of interest. The information to be recovered after decoding is only the relevant part of the source data. For example, one might have access to the transcript of an audio recording which has an entropy by orders of magnitude lower than the original audio data.
The transcript here can be used as side information to compress the audio data further than what  can be achieved by RDT, without increasing the distortion of the relevant information. 

Loss of information in the context of a compression-decompression scheme is mathematically equivalent to transmission of information through a noisy channel. Rather than characterising the information lost by encoding, one characterises the information lost during transmission. The Information Bottleneck Method is formulated as a communication problem with the relevant variable acting as side information.  Iterative algorithms to compute the optimum channel achieving the task are also provided in \cite{Tishby1999}.

In \cite{Still2016}, a first quantum extension of the Information Bottleneck Method was developed and applied to predictive filtering. 
In this paper we extend the Information Bottleneck Method to the quantum case by considering the transmission of quantum information through a quantum channel with side information. 
We derive the compression rate for an entanglement-assisted protocol with classical communication for our quantum extension of the IBM. Our derivation relies on the conjecture that the rate function is convex in distortion.
The optimum quantum channel that can achieve this rate is also characterised. 

A quantum extension to RDT was introduced by Barnum \cite{Barnum2000}. However, the results were unsatisfactory since the bound on the rate was given in terms of coherent information which can be negative. The results were improved and a complete characterisation of quantum channels achieving rate-distortion coding in a number of settings 
was given by Datta \emph{et al.} \cite{Datta2013}. Various settings of quantum RDT in the presence of auxiliary information were discussed in the work of Wilde \emph{et al.} \cite{Wilde2013}. However, the specific question of transmitting relevant information asked in the IBM with its choice of distance measure and the specifics of the optimisation problem have not been considered yet. 

The setting of the classical IBM is as follows. A correlated pair of random variables $X$ and $Y$ is given with a joint probability distribution $P(X,Y)$. The task is to find the optimum channel with input $X$ and output $\tilde X$ such that $\tilde X$ retains a fixed amount of correlation, $C$, with variable $Y$. The amount of correlation is quantified by the Shannon mutual information $I(\tilde X;Y):=H(\tilde X)+ H(Y)-H(\tilde XY)$, where $H(\cdot)$ is the Shannon entropy. For a successful completion of the IBM task it is required that $I(\tilde X;Y)\geq C$. 
Representing the channel by the conditional probability distribution $P(\tilde X |X)$, one can show that the classical rate of compression for a given minimum amount of correlation $C$, $R_{\text cls}(C)$,  is given by 
\EQ{
R_{\text cls}(C)=\min_{P({\tilde X}|X):I(\tilde X;Y)\geq C} I(X;\tilde X) \label{cIBM},
}
as it was first proposed in \cite{Tishby1999} and proved in \cite{BNT03}.

The IBM, however, is concerned not with an output distribution close to the input distribution but with an output characterised by its information about some other variable $Y$. The task of the IBM is to find 
the lowest value of $I(X;\tilde X)$ such that $I(\tilde X; Y)$ is still above some given threshold. The value of $I(X; \tilde X)$ can be reinterpreted as a communication rate, namely the number of transmitted bits needed to specify an element of $\tilde X$, per element of $X$ \cite[Sec.~2]{Tishby1999}. Minimising the mutual information with respect to all  channels that satisfy the threshold criterion achieves the task.\footnote{Note that the analogy between IBM and RDT is in spirit only. The technical difference is the distortion measure, which is a function of the output alphabet in the case of RDT while it is a function of the probability distribution of the outputs in the case of IBM.}

The channel that achieves the rate in Eq.~(\ref{cIBM}) can be found by the Lagrangian technique. The Lagrangian is defined as 
\EQ{
\mathcal{L}_{cls}= I(X;\tilde X)-\beta I(\tilde X;Y)-\sum_{x,\tilde x}\lambda(x)P(\tilde x|x), \label{clagrange}
}
where $\beta$ is the Lagrange multiplier for the information constraint and $\lambda(x)$ are the Lagrange multipliers for the normalisation constraint of the conditional distribution $P(\tilde X|X)$. Taking the derivative of the Lagrangian with respect to the channel and setting it to $0$ gives the expression for the channel as 
\EQ{
P({\tilde x}|x)=P({\tilde x})\frac{e^{-\beta D(P(Y|x)||P(Y|\tilde x))}}{Z}.\label{clopt}
}
$D(.||.)$ in the exponent on the right hand side is the Kullback-Leibler divergence of two probability distributions $P(Y|x)$ and $P(Y|\tilde x)$, and $Z$ is the normalising factor.

The  setting of the quantum IBM is as follows.  The input to the channel is the system $X$ which is described by a state $\rho_{XY}$ with side information $Y$. The channel acts on the $X$ part of this state,  $\rho_X$. The output of the quantum channel is the system $\tilde X$ which is also an entangled state with the side information $Y$, $\tilde\rho_{\tilde XY}$. 
Correlations in the state $\tilde\rho_{\tilde XY}$ are measured by the  quantum mutual information $I(\tilde X;Y)_{\tilde\rho}:=S(\tilde X)_{\tilde\rho}+S(Y)_{\tilde\rho}-S(\tilde XY)_{\tilde\rho}$ where $S(\cdot)$ is the von Neumann entropy to base $e$.
The Bottleneck constraint in the quantum case is  $I(\tilde X;Y)_{\tilde\rho}\geq J$, that is, a minimum amount of correlation $J$ between systems $\tilde X$ and $Y$ must be retained.
This choice of measure for correlation between $\tilde X$ and $Y$ naturally generalizes the setup of the classical Bottleneck method; it includes both classical correlation and quantum entanglement between the two systems.

In the next section we provide the  protocol for the quantum IBM which results in a conjecture for the quantum counterpart of the compression rate in Eq.~(\ref{cIBM}).  Section \ref{lagra} contains the channel optimisation,
resulting in the quantum counterpart of Eq.~(\ref{clopt}).

\section{A  Protocol for the Quantum Information Bottleneck Method \label{howto}}

In this section we explicitly describe a protocol that simulates the action of the channel in the setting discussed in the introduction by compression and decompression channels and classical communication with the assistance of entangled resources. We then define the compression rate for this protocol, before proving our main result. 

Following the usual approach in RDT we call a rate-distortion pair $(r,J)$ achievable,  
where $r\geq 0$ is the rate and $J \geq 0$ is the distortion, if, roughly speaking, a lossy channel exists such that it can transmit a message, i.e., a state of the given input system, by sending $r$ bits with at the amount of distortion not exceeding $J$. The rate function $R(J)$ is then given as an infimum over achievable pairs. Of course, it is crucial to specify what exactly is meant by ``lossy channel'' and ``achievable'' in our context, and this is what we will discuss now.

\begin{figure}[!t]
  \centering
      \includegraphics[width=0.45\textwidth]{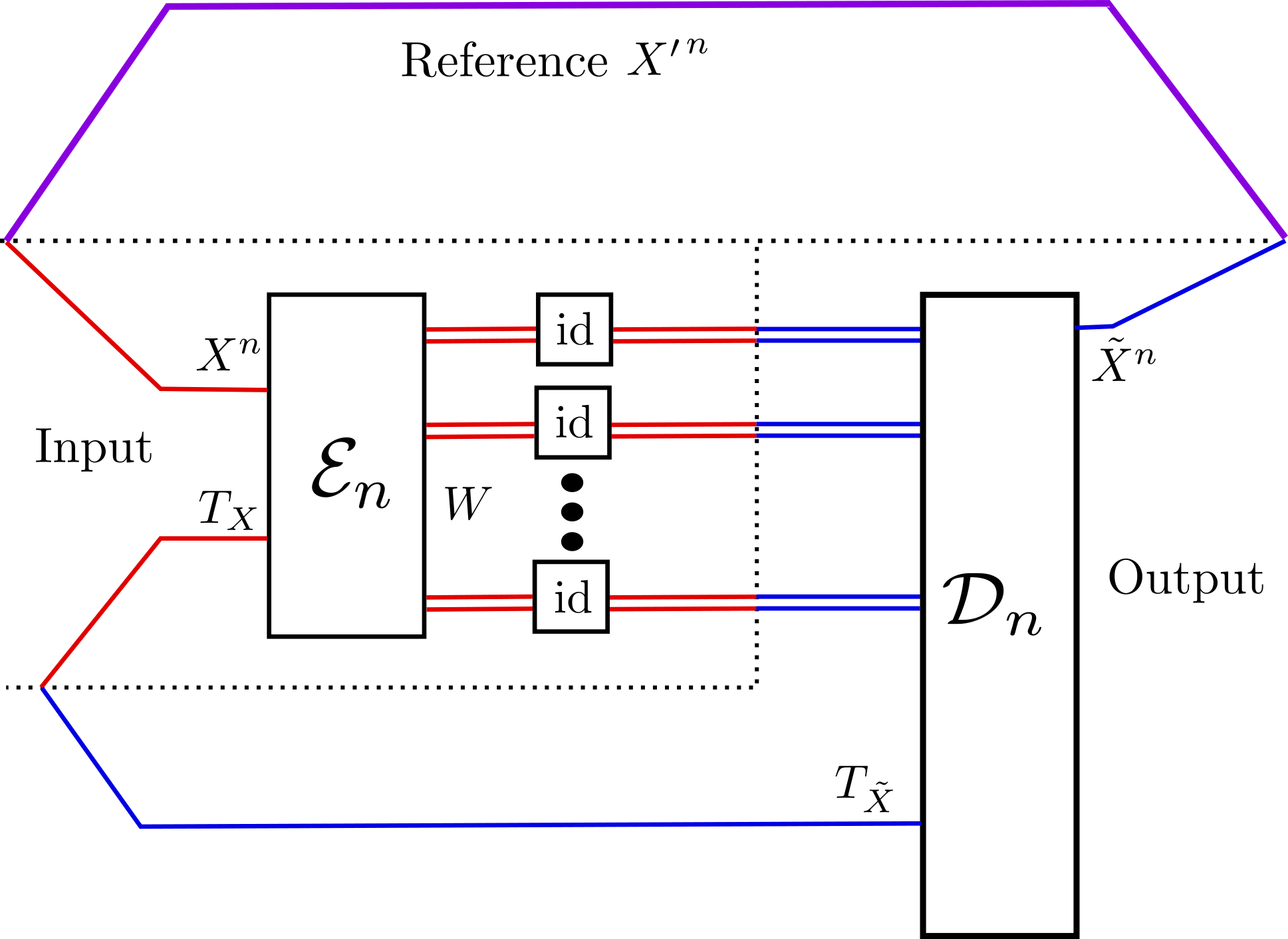}
 \caption{\small The protocol for the entanglement assisted Quantum Information Bottleneck Method. The compression encoding $\mathcal{E}_n$ acts on
  $n$ copies of the state of the system $X$. The $n$ copies of the state of the system $Y$ are used as the \emph{relevance variable} and are entangled with 
  the system $X$. However, as the system $Y$ is not transmitted through the channel, we do not depict the system in the figure. $T_{X}$ and $T_{\tilde X}$ 
  are the two entangled systems 
  that the input and the output of the protocol share to assist the transmission. The output of the compressing map $\mathcal{E}_n$ is classical system $W$ which is transferred to the output section for 
  decompression, $\mathcal{D}_n$, via 
 the noiseless channels ``id". The reference, labeled $X'$, is what purifies the state $\rho_{X}$. The state $\rho_{X}$ is the reduced density operator of 
 the given initial state $\rho_{XY}$.}
  \label{diagram}
\end{figure}

Let us consider the entanglement-assisted protocol with noiseless classical communication illustrated in Fig.~\ref{diagram}. The information
of system $X$ to be compressed is represented by $n$ independently and 
identically distributed (i.i.d.) copies, $\rho_{X}^{\otimes n}$, of the density operator $\rho_{X}$, with a purification $\tau_{X'X}$.  The input, however, is correlated with a system $Y$ which contains our relevant information, and the state is denoted $\rho_{XY}$.
The sender and receiver of the information
share an entangled resource $\Phi_{T_X T_{\tilde X}}$, where the system $T_X$ is with the sender and the system $T_{\tilde X}$ is with the receiver. The sender acts on the 
input state, $\rho_X^{\otimes n}$, and the state of half of the entangled pair, $T_X$, with the compression map $\mathcal{E}_n:=\mathcal{E}_{X^n T_X \to W}$, 
where $W$ is a classical system of size at most $ e^{nr}$, and $r$ is the communication rate. Then, the receiver acts on $W$ 
and the state of the other half of the 
entangled pair, $T_{\tilde X}$, with the decompression channel $\mathcal{D}_n:=\mathcal{D}_{W T_{\tilde X} \to {\tilde X}^n}$. 
The overall action $\mathcal{F}_n := \mathcal{D}_n \circ \mathcal{E}_n$ of the compression-decompression channel, as specified above, defines an \emph{$(n,r)$ quantum rate-distortion code}. 

Given such $\mathcal{F}_n$, we consider the marginal operation, $1 \leq i \leq n$,
\EQ{
&\mathcal{F}_n^{(i)}(\xi_X):=\\
&\operatorname{Tr}_{\tilde X_1,...,\tilde X_{i-1}, \tilde X_{i+1},..., \tilde X_n}[\mathcal{F}_n(\rho_{X}^{\otimes (i-1)} \otimes \xi_X \otimes \rho_{X}^{\otimes (n-i)})].
}
Then, for any $i$, we can define 
\EQ{
\tilde\sigma_{\tilde X_i  Y_i} & :=& \mathcal{F}_n^{(i)}\otimes \mathcal{I}_{Y} (\rho_{XY}), \label{03appmarg}
}
where $\mathcal{I}_{Y}$ is the identity channel acting on the system $Y$, and its partial traces are
\EQ{ 
\tilde\sigma_{\tilde X_i} = \mathcal{F}^{(i)}_n (\rho_X),\quad 
\tilde\sigma_{Y_i}  = \rho_Y.\label{03appsigma1}
}
Using Eqs.~\eqref{03appmarg} and \eqref{03appsigma1} we can define the mutual information
\EQ{
I_i(\tilde X; Y) & := S(\tilde\sigma_{\tilde X_i}) + S(\tilde\sigma_{Y_i}) - S(\tilde\sigma_{\tilde X_i  Y_i}),\label{03appmutu}
}
and its average over many uses of the channel,
\EQ{\label{barI}
\bar{I}_n(\tilde X; Y) := \frac{1}{n} \sum_{i=1}^n I_i(\tilde X; Y).
}
Now, for any $r, J \geq 0$, we call $(r,J)$ an \emph{achievable rate-distortion
pair} (for the Quantum Information Bottleneck Method) if there exists a sequence of $(n,r)$ quantum rate-distortion codes such that 
\EQ{
\lim_{n\to \infty}  \bar{I}_n(\tilde X; Y) >J. \label{discr}
}

Finally, the rate function $R(J)$ is defined as 
\EQ{
R(J):=\inf \{r:(r,J) \mbox{ is achievable} \}. \label{03ratedef}
}

Given these definitions,  we can now express $R(J)$ in a more direct way in terms of \emph{one} copy of the input state $\rho_{XY}$.

\begin{conjecture}
 Consider a memoryless quantum information source $\rho_{XY}$, where system $X$ is to be compressed such that after decompression it retains an amount of correlation $0\leq J \leq 1$ with system $Y$. The asymptotic Bottleneck rate 
 function for entanglement-assisted lossy source coding with noiseless classical communication is given by 
 \EQ{
R(J)=\min_{\mathcal{N}_{X \to \tilde X}:I(\tilde X;Y)_{ \tilde\rho}\geq J} I(X';\tilde X)_{\tilde\tau}, \label{qrate}
}
under the  conjecture that the expression
on the right hand side of Eq.~\eqref{qrate} is convex in $J$. Here, the  quantum mutual information $I(X';\tilde X)$ is evaluated over the state 
\EQ{
\tilde\tau_{X'\tilde X}:=(\mathcal{I}_{X'}\otimes\mathcal{N}_{X\to\tilde X})(\tau_{X'X}),\label{tau}
}
such that $\tau_{X'X}$ is the purification of the reduced state $\rho_X$ of the state $\rho_{XY}$, and 
\EQ{
\tilde{\rho}_{\tilde{X}Y}:=(\mathcal{N}_{X\to\tilde X}\otimes\mathcal{I}_{Y})(\rho_{XY})}.
\end{conjecture}
While we are currently unable to remove the extra convexity assumption, we have verified it in numerical examples which are discussed in Appendix~\ref{numerics}.
\begin{proof}
Here we follow the approach of \cite[Theorem 2]{Datta2013}. We temporarily denote the right-hand side of Eq.~(\ref{qrate}) as
\begin{equation}\label{M}
M(J) := \min_{\mathcal{N}_{X \to \tilde X}:I(\tilde X;Y)_{\tilde\rho}\geq J} I(X';\tilde X)_{\tilde\tau}.
\end{equation}
%
We need to show achievability of the rate , $R(J) \leq M(J)$, as well as optimality, $R(J)\geq M(J)$. We start with optimality. Let $(r,J)$ be an achievable rate-distortion pair and $\mathcal{F}_n$ a corresponding sequence of codes.
We have for large $n$,
\EQ{
 nr &\geq  S(W) \\
 &\geq S(W|T_{\tilde X})\\
 &\geq S(W|T_{\tilde X})-S(W|X'^nT_{\tilde X})\\
 & = I(W;X'^n|T_{\tilde X})\\
 &=I(W;X'^n|T_{\tilde X})+ I(X'^n;T_{\tilde X})\\
 &=I(WT_{\tilde X};X'^n)\\
 &\geq I({\tilde X}^n;X'^n)\\
 &\geq \sum_{i=1}^n I({\tilde X}_i;X'_i)\\
 &\geq \sum_{i=1}^n M(I_i({\tilde X};{ Y}))\\
 &= n\sum_{i=1}^n \frac{1}{n} M(I_i({\tilde X};{ Y}))\\
 &\geq n M( \bar{I}_n(\tilde X;Y))\\
 &\geq n M(J). \label{ineq}
}
The first inequality follows from the fact that the entropy of the uniform distribution, $nr$, is the upper bound of $S(W)$. The second inequality follows because entropy is nondecreasing under conditioning. The third inequality follows because the state of the system $WX'^n T_{\tilde X}$ is separable with respect to the classical system $W$ and therefore $S(W|X'^n T_{\tilde X}) \geq 0$ \cite[footnote 10]{Cerf1997}. The first equality follows from the definition of mutual information. The second equality follows since the state $\Phi_{T_X T_{\tilde X}}$ is in a tensor product with the state of the remaining input and therefore $I(X'^n; T_{\tilde X}) =0$.  
In the third equality we use the chain rule for mutual information. The fourth inequality follows from the data processing inequality.
The fifth inequality follows from superadditivity of quantum mutual information \cite[Lemma 15]{Datta2013}. The sixth inequality follows from the definition of $M(J)$, where we use the channel $\mathcal{N}_{X \rightarrow \tilde X} = \mathcal{F}_n^{(i)}$.
In the seventh inequality we have used  conjecture of convexity of $M$. 
Finally the last inequality follows for large $n$ from Eq.~(\ref{discr}), using that $M(J)$ is a nondecreasing
function of $J$. The rate function is nondecreasing in $J$, because for any $J'>J$ the domain of minimisation in Eq.~(\ref{qrate}) becomes smaller, which implies that the rate function can only become larger. Thus, since $(r,J)$ was arbitrary, (\ref{ineq}) implies $R(J) \geq M(J)$.

Achievability follows from an application of the Quantum Reverse Shannon Theorem (QRST) \cite{Bennett2014,Berta2011,Devetak2006}. The general form of the QRST states
that a quantum channel can be simulated by an unlimited amount of shared entanglement and an amount of classical
communication equal to the channel's entanglement-assisted classical capacity. Specifically, in our case, fix $J >0$ and let
$\mathcal{N}_{X \rightarrow \tilde X}$ be the optimum channel at which the minimum in Eq.~(\ref{M}) is attained. For a given $\epsilon >0$, we use the QRST to construct a sequence of channels 
$\mathcal{F}_n = \mathcal{D}_n \circ \mathcal{E}_n$ such that they are close to an $n$-fold application of $\mathcal{N}_{X \rightarrow \tilde X}$, in the sense that 
\EQ{ \label{sigman}
\| \tilde\sigma_{X'^n \tilde X^n } - \tilde\tau^{\otimes n}_{X' \tilde X} \|_1 \leq \epsilon,
}
where $\tau_{X' \tilde X}$ is defined in Eq.~(\ref{tau}) and
\begin{equation}\label{sigma}
\tilde\sigma_{X'^n \tilde X^n}:=\mathcal{I}_{X'^n}\otimes \mathcal{F}_n (\tau_{X'X}^{\otimes n}).
\end{equation}
According to the QRST we can construct such a sequence of channels, $\mathcal{F}_n$, using classical communication at rate $r= I(X'; \tilde X)_{\tilde\tau}$, given an unlimited 
amount of entanglement for a known tensor power input, which is our setting, \emph{c.f.} \cite[Fig.2, top right]{Bennett2014}\footnote{Notice that the QRST in Fig. 2 of \cite{Bennett2014} was done in two ways, a so-called
\emph{feedback} and a \emph{non-feedback} simulation, respectively. The former means that the environment of the simulated channel is in the possession of the sender, while the latter means some part of the environment of the channel 
will end up on the receiver's side. In both cases, if unlimited entanglement is available, the rate is given by the single-letter formula,  the above mutual information. In the case of the non-feedback simulation of 
the channel $\mathcal{N}_{X \to \tilde X}$, the feedback
simulation protocol \cite[Theorem 3(a)]{Bennett2014} is used, namely the isometry $\mathcal{U}_\mathcal{N}^{X\to E\tilde X}$ composed with the isometry $V^{E^n \to E_A E_B}$, is simulated by QRST, such that systems 
$E_A$ and $X$ are with the sender and the systems $E_B$ and $\tilde X$ are with the receiver. Tracing out the environment from this channel
gives the desired channel, $\mathcal{N}_{X\to\tilde X}$. The result listed in Fig. 2 of \cite{Bennett2014} can be obtained from Theorem~3(b) there as follows. In the presence of an unlimited shared entanglement, 
condition (21) in \cite{Bennett2014} trivially holds and the maximum in Eq.~(23) there is always attained by the first of the two terms. To obtain the desired expression of the mutual information from (20) and (23) 
there, it is enough to choose $E^n, E_A$ and $E_B$ such that the isometry $V$ acts trivially and we have $I(R^n:B^n E_B)= I(R^n:B^n)$. Finally, the single letter formula $I(R:B)$ that we use follows from the fact that 
we have a known tensor-power input.}
\footnote{Notice that the rates shown in Fig. 2 of \cite{Bennett2014} are quantum communication rates which is why they are different by a factor of 1/2 from what we use. The extra factor drops because 1-qubit channels
can be simulated with 2 bit classical channels in the presence of entanglement using the teleportation protocol.}.
From Eq.~\eqref{sigman} and the fact that $I(\tilde X;Y)\geq J$ for the channel $\mathcal{N}_{X \rightarrow \tilde X}$, one can show (see Lemma~\ref{lemmaB} in Appendix~\ref{average}) that Eq.~\eqref{discr} is fulfilled with
$J- \delta$ instead of $J$, where $\delta \rightarrow 0$ as $\epsilon \rightarrow 0$. Hence, $(r, J - \delta)$ is achievable and $R(J - \delta) \leq M(J)$. 

We have now shown that
\begin{equation}
\forall \delta >0 \; :\; M(J - \delta) \leq R(J- \delta) \leq M(J).
\end{equation}
From this it follows that
\begin{equation}
\lim_{\delta \searrow 0} R(J -\delta) = M(J).
\end{equation}
Since $R$ is nondecreasing, by Eqs.~\eqref{discr} and \eqref{03ratedef}, and $M$ is continuous (a property that follows from the conjecture of convexity), $M(J) = R(J)$ for all $J$. 

\end{proof}

\emph{Remark 1}: In Eq.~(\ref{qrate}) the Quantum Reverse Shannon Theorem is used, which in addition uses shared entanglement (``entanglement assistance'') to generate a protocol that achieves the rate. However, the requirements for the Quantum Reverse Shannon Theorem are much more stringent than 
those of the Bottleneck method. Therefore, it might be possible to find a rate function without entanglement assistance. This is still an open problem.

{\emph{Remark 2}: The systems $XY$ and $XX'$ are understood to be in separate setups; the states $\tilde\rho_{\tilde X Y}$ and $\tilde\tau_{X' \tilde X}$ are not partial traces of a common parent state,  but they involve an identical channel $\mathcal{N}$ (``two separate experiments using the same apparatus'').
The QRST applies only to the system $XX'$. To prove Eq.~\eqref{qrate}, { these two setups need to be connected}, which is done in Lemma~\ref{lemmaB} in Appendix~\ref{average}.}

\emph{Remark 3}: Generally, rate distortion functions obtained from the definition in Eq.~(\ref{03ratedef}) are non-increasing functions of the distortion,
whereas the RHSs of Eqs.~(\ref{qrate}) and (\ref{cIBM}) are both non-decreasing functions. This is not a fundamental difference between RDT and the Bottleneck method. It is merely due to the fact that in the IBM the constraint of minimisation is chosen to be the amount of correlation preserved, $I(\tilde X;Y)\geq J$, while in RDT  the constraint of minimisation is the average loss of information, $\langle d(X;\tilde X) \rangle\leq D$, for some fixed $D$ and some distortion measure $d(X;\tilde X)$, characterising the noise introduced by a channel $\mathcal{N}_{X \to \tilde X}$ to a state $\rho_X$. 

The formulation of the IBM can easily be changed to using a  constraint on the loss of correlation such as $I(X;Y)-I(\tilde X;Y)\leq D$, 
in which case the rate function is a non-increasing function of $D$. Rather than changing it, the  structure of the minimisation constraint is kept in line 
with the classical IBM. 

\emph{Remark 4}: As  discussed above, in a lossy compression-decompression protocol the minimisation is performed over all  channels  satisfying a certain criterion, see 
Eq.~(\ref{qrate}).
Although the functions $I(\tilde X; X')$ and $I(\tilde X; Y)$ are convex in the channel, the optimisation problem in Eq.~(\ref{qrate}) is not of convex type due to the sign of the inequality in the constraint ($\geq$ rather 
than $\leq$). Just as in the classical case, it is a so-called ``reverse convex problem'' (see, e.g., \cite{Singer07}), for which many of the standard results, such as strong duality or convexity of the resulting
function in $J$, are not known to hold or apply only in a weaker form. Nevertheless, we present numerical evidence for convexity of the RHS in Appendix~\ref{numerics}. We also show that the minimum in 
Eq.~\eqref{qrate} is actually attained at the boundary $I(\tilde X; Y)_{\tilde\rho} = J$ (see Lemma~\ref{lemma:boundarY} in Appendix~\ref{BA}), which makes it useful to search for local optima with the method of
Lagrange multipliers, as we show in the next section.

\section{The Optimisation Problem for the Quantum Information Bottleneck Method\label{lagra}}

We now proceed to the optimisation problem.
In order to formulate the Lagrangian corresponding to the quantum counterpart of Eq.~(\ref{clagrange}), we need to choose a suitable way of representing the channel $\mathcal{N}$, e.g., the Kraus operators or the Choi-Jamio\l{}kowski representation. 
It turns out that indeed the most compact and convenient way to compute the derivatives of the Lagrangian is with respect to the Choi-Jamio\l{}kowski representation defined as
\EQ{
\Psi_{X' \tilde X}:=\big( \mathcal{I}_{X'} \otimes \mathcal{N}_{X \to {\tilde X} }\big) (\tilde \Phi_{X'X}),
}
where $\tilde \Phi_{X'X}:= \sum_{i,j=0}^{d-1}\ketbra{i}{j}_{X'} \otimes \ketbra{i}{j}_{X}$. 
For the rate function given in Eq.~(\ref{qrate}) one can write the Lagrangian
\EQ{
\mathcal{L}&:=& I(X';\tilde X)_{\tilde\tau} -\beta I(\tilde X; Y)_{\tilde\rho} -\operatorname{Tr}_{X \tilde X}(\Psi_{X \tilde X}^{t_{X}}(\Lambda_{X} \otimes I_{\tilde X} )), \label{lag}
}
{where $t_X$ denotes the partial transpose with respect to the basis $\{\ket{i}_X\}_{i}$ on $X$;} here $\beta$ is the Lagrange multiplier for the constraint of minimisation and the Hermitian operator $\Lambda_{X}$ is the Lagrange multiplier to guarantee that the channel is a 
completely positive trace preserving map. 
The states  in Eq.~(\ref{lag}) can be written as  functions of the Choi-Jamio\l{}kowski state of the channel  $\Psi_{X\tilde X}$. 
The joint state $\tau_{X' X}$ in Eq.~(\ref{lag}) can be written as 
\EQ{
\tilde\tau_{X'\tilde X}&=  \operatorname{Tr}_X \big\{ \Psi_{X \tilde X}^{t_X} \tau_{X' X}\big\}\\
&= (\rho^X_{X'}\otimes I_{\tilde X})^{1/2}\Psi_{X' \tilde X}^{t_{X'}}(\rho^X_{X'}\otimes I_{\tilde X})^{1/2}, \label{03choijoint}
}
where $\rho^X_{X'}$ is the same state as $\rho_X$ acting on the Hilbert space $\mathcal{H}_{X'}$ of the system $X'$.

By similar considerations, one can show that the joint state \begin{equation}\label{rho}
\tilde\rho_{\tilde XY}:=(\mathcal{N}_{X \to \tilde X} \otimes \mathcal{I}_y) (\rho_{XY})
\end{equation}
 can be written as 
\EQ{
\tilde\rho_{\tilde X Y}=&\operatorname{Tr}_{X}\big\{\Psi_{X \tilde X}^{t_{X}}\rho_{XY}\big\}
\nonumber \\
=& \operatorname{Tr}_{X'Y'}\Big\{ (\rho^{XY}_{X'Y'}\otimes I_{\tilde X  Y})^{1/2}(\Psi_{X' \tilde X}^{t_{X'}} \otimes \Phi^{t_{Y'}}_{Y'  Y})\\
&(\rho^{XY}_{X'Y'}\otimes I_{\tilde X  Y})^{1/2}\Big\}, \label{rhotildeXY}
}
where we have chosen a maximally entangled state $ \Phi_{Y'  Y}$ and a corresponding transpose $t_{Y'}$.
In the third term of Eq.~\eqref{lag} the dependence on the channel state is already explicit.

Let
\EQ{
D_{X \tilde X}^{\beta Y}:=&
\beta \log \tilde\rho_{\tilde X} \otimes I_{X}\\
&-\beta\operatorname{Tr}_{Y}\big\{(\rho_{X}^{-1/2} \rho_{X Y}\rho_{X}^{-1/2} (\log \tilde\rho_{\tilde X Y}\otimes I_{X }) \big\},\label{03qdis}
}
and 
\EQ{
\tilde\Lambda_{X} :=\rho_{X}^{-1/2} \Lambda_{X}\rho_{X}^{-1/2}.
}
Taking the derivative of the Lagrangian in Eq.~(\ref{lag}) with respect to the channel and setting it
to zero (for details, see Appendix~\ref{BA}) gives the optimum channel as 
\EQ{\label{sol}
\Psi_{X \tilde X }^{t_{X}}=& (\rho_{X} \otimes I_{\tilde X })^{-1/2}\\
&e^{ \log \tilde\rho_{\tilde X}\otimes  I_{X}-D_{X \tilde X}^{\beta Y}+\tilde \Lambda_{X} \otimes I_{\tilde X}}(\rho_{X} \otimes I_{\tilde X })^{-1/2}.
}
Note that this determines $\Psi_{X \tilde X}^{t_{X}}$ implicitly since it also appears on the right hand side of this equation in $\tilde\rho_{\tilde X}$ and in the definition of $D_{X \tilde X}^{\beta Y}$ (for details, see Appendix~\ref{BA}). {In order to find the optimum channel, {Eq.~\eqref{sol} needs to be solved iteratively} for $\Psi_{X\tilde X}^{t_{X}}$. 
The solution also depends on the unknown Lagrange multipliers $\Lambda_{X}, \beta$ associated with the constraints of the problem; these need to be determined in a further optimisation step. We comment on a possible algorithm to that end in Appendix~\ref{BA}.}

Eq.~(\ref{sol}) reduces to its classical counterpart in Eq.~(\ref{clopt}) in the case of diagonal density operators. 
To see this, consider the diagonal case where the density operators reduce to probability distributions. From Eq.~(\ref{sol}) it follows that 
\EQ{
P(\tilde x|x)=&\frac{1}{P(x)}\exp\Big\{ \log P(\tilde x)-\beta \big( \log P(\tilde x)\\
&-\sum_y P(y|x)\log P(\tilde x y)\big)+\frac{\lambda(x)}{P(x)}\Big\}, \label{03reduction}
}
with $\lambda(x)$ being the same normalisation Lagrange multiplier as in Eq.~(\ref{clagrange}). Notice that since $H(Y|X=x)=-\sum_y P(y|x)\log P(y|x)$ depends 
only on $x$ but not on $\tilde x$, it can be absorbed into  $\lambda(x)$. Defining
\EQ{
\tilde \lambda(x):=\frac{\lambda(x)}{P(x)}-\beta H(Y|X=x)-\log P(x)~,
}
Eq.~(\ref{03reduction}) becomes
\EQ{
P({\tilde x}|x)=P({\tilde x})e^{-\beta D(P(Y|x)||P(Y|\tilde x))+\tilde \lambda(x)}, 
}
which is the same classical channel as Eq.~(\ref{clopt}), with all the extra terms being absorbed into the normalisation factor. This also shows that 
$D_{X\tilde X}^{\beta y }$ is a quantum operator corresponding to the distance measure in the classical Bottleneck method. The idea of distance operators has been 
used in a number of quantum information processing tasks \cite{Wilde2013a,Datta2013a}, however the $D_{X \tilde X}^{\beta Y}$ is particular to the present setting.
Eq.~(\ref{sol}) can be 
used in principle to compute numerical values of quantum channels using iterative algorithms, akin to their classical counterparts by methods 
introduced by Blahut and Arimoto \cite{Blahut1972, Arimoto1972}.

Finally, we would like to remark that a previous result in the domain of quantum information theory was inspired by the classical Information Bottleneck Method \cite{Still2016}, and indeed some of the technical steps in Section III of our paper are very similar. Here we clarify the difference between our contribution and that of \cite{Still2016}. To reflect on our contribution, consider that in classical communication, Shannon's channel capacity theorem gives a single-letter formula that fully characterises a communication channel. This allows one to interpret the mutual information between the input and the output of a channel as a communication rate, as was done in~\cite{Tishby1999}. In quantum communication theory the situation is different since quantum channels (except for very few cases) are typically not characterised by such single-letter quantities. Therefore, in order to give a communication theoretic interpretation to an entropic quantity such as the one that we (and \cite{Still2016}) use, one has to prove that the quantity corresponds to a rate in a communication scenario by giving a coding scenario to which the rate refers. We have done  this, up to the previously stated conjecture, by designing a setup that is compatible with the Quantum Reverse Shannon Theorem (notice, e.g., the unlimited amount of entanglement), and by proving the lemma given in Appendix A. This lemma shows that in our setting the Quantum Reverse Shannon Theorem can be invoked. In Ref. \cite{Still2016}, the focus there was on applying a quantum extension of the Information Bottleneck Method to predictive filtering, whereas we focus here on the communication scenario.

\section{Conclusion and Outlook}

This paper  introduced the quantum extension of the Information Bottleneck Method. This method compresses the data such that only the information relevant with respect to some given variable is preserved. We derive a lower bound to the compression rate of relevant quantum information. The problem was formulated as a communication channel problem and the rate was shown to be achievable by explicitly constructing a channel achieving it. Just like in the classical case, the compression rate of the quantum Information Bottleneck Method is lower than that given by quantum rate distortion theory. 
Several conceptual issues arose from the structural differences between the mathematical formalism of quantum theory and 
classical probability theory which were discussed and solutions were presented.

Some open questions remain. Our proof of Eq.~\eqref{qrate} relied on an unproven conjecture (convexity of the expression on the right hand side of Eq.~\eqref{qrate} in $J$). While this seems to be fulfilled in examples (cf. Appendix~\ref{numerics}), a proof of this property is currently missing.

{In Appendix~\ref{numerics} a simple algorithm is used to compute the optimum channel and thus the rate function $R(J)$ in low dimensional systems; but for systems of realistic size a more efficient algorithm would be required. This might be based on numerically solving the implicit equation \eqref{sol}.} 

\section*{Acknowledgment}
The authors are grateful to Henning Bostelmann for discussions on the subject and for assistance with the numerics in Appendix~\ref{numerics}. Further, D.~C.~would like to thank Roger Colbeck and Robert Koenig for helpful discussions.
P.~K.~would like to thank David Sutter and Renato Renner for comments on earlier versions of the manuscript.

S.~S.~is supported by the National Science Foundation of China through Grant No.~11675136, the Hong Kong Research Grant Council through Grant No.~17326616, the EPSRC National Quantum Technology Hub in Networked Quantum Information Technologies. This publication was made possible through the support of a grant from the John Templeton Foundation. The opinions expressed in this publication are those of the
authors and do not necessarily reflect the views of the John Templeton Foundation.
P.~K.~acknowledges funding from the Swiss National Science Foundation (through the National Centre of Competence in Research ``Quantum Science and Technology'' and through Project No. 200020\_165843), from the European Research Council (grant 258932) and through COST MP1209 "Thermodynamics in the quantum regime".

K.W. acknowledges support through EPSRC grant I013717/1.

{\small


\begin{thebibliography}{10}

\bibitem{code}
The code is supplied with this article as a Maple worksheet {\tt
  randomchannels.mw} and a PDF copy as {\tt randomchannels.pdf}.

\bibitem{Arimoto1972}
Suguru Arimoto.
\newblock An algorithm for computing the capacity of arbitrary discrete
  memoryless channels.
\newblock {\em IEEE Transactions on Information Theory}, 18(1):14--20, 1972.

\bibitem{Barnum2000}
Howard Barnum.
\newblock Quantum rate-distortion coding.
\newblock {\em Physical Review A}, 62(4):042309, 2000.

\bibitem{Bennett2014}
Charles~H Bennett, Igor Devetak, Aram~W Harrow, Peter~W Shor, and Andreas
  Winter.
\newblock The quantum reverse shannon theorem and resource tradeoffs for
  simulating quantum channels.
\newblock {\em IEEE Transactions on Information Theory}, 60(5):2926--2959,
  2014.

\bibitem{Berta2011}
Mario Berta, Matthias Christandl, and Renato Renner.
\newblock The quantum reverse shannon theorem based on one-shot information
  theory.
\newblock {\em Communications in Mathematical Physics}, 306(3):579, 2011.

\bibitem{Blahut1972}
Richard~E Blahut.
\newblock Computation of channel capacity and rate-distortion functions.
\newblock {\em IEEE Transactions on Information Theory}, 18(4):460--473, 1972.

\bibitem{Cerf1997}
N.~J. Cerf and C.~Adami.
\newblock Negative entropy and information in quantum mechanics.
\newblock {\em Physical Review Letters}, 79(26):5194--5197, 1997.

\bibitem{Datta2013}
Nilanjana Datta, Min-Hsiu Hsieh, and Mark~M Wilde.
\newblock Quantum rate distortion, reverse shannon theorems, and source-channel
  separation.
\newblock {\em IEEE Transactions on Information Theory}, 59(1):615--630, 2013.

\bibitem{Datta2013a}
Nilanjana Datta, Joseph~M Renes, Renato Renner, and Mark~M Wilde.
\newblock One-shot lossy quantum data compression.
\newblock {\em IEEE Transactions on Information Theory}, 59(12):8057--8076,
  2013.

\bibitem{Devetak2006}
Igor Devetak.
\newblock Triangle of dualities between quantum communication protocols.
\newblock {\em Physical review letters}, 97(14):140503, 2006.

\bibitem{BNT03}
Ran Gilad-Bachrach, Amir Navot, and Naftali Tishby.
\newblock An information theoretic tradeoff between complexity and accuracy.
\newblock In Bernhard Sch{\"o}lkopf and Manfred~K. Warmuth, editors, {\em
  Learning Theory and Kernel Machines}, pages 595--609, Berlin, Heidelberg,
  2003. Springer Berlin Heidelberg.

\bibitem{Still2016}
A.~L. Grimsmo and S.~Still.
\newblock Negative entropy and information in quantum mechanics.
\newblock {\em Physical Review A}, 94:012338, 2016.

\bibitem{Shannon1948}
Claude~E Shannon.
\newblock A mathematical theory of communication.
\newblock {\em Bell System Technical Journal}, 27:379--423, 1948.

\bibitem{Shannon1959}
Claude~E Shannon.
\newblock Coding theorems for a discrete source with a fidelity criterion.
\newblock {\em IRE Nat. Conv. Rec}, 4(142-163):1, 1959.

\bibitem{Singer07}
Ivan Singer.
\newblock {\em Duality for Nonconvex Approximation and Optimization}.
\newblock Springer Science \& Business Media, 2007.

\bibitem{Tishby1999}
N~Tishby, F~C Pereira, and W~Bialek.
\newblock The information bottleneck method.
\newblock {\em Proceedings of the 37th Annual Allerton Conference on
  Communication, Control and Computing}, pages 368--377, 1999.

\bibitem{Wilde2013a}
Mark~M Wilde.
\newblock {\em Quantum information theory}.
\newblock Cambridge University Press, 2013.

\bibitem{Wilde2013}
Mark~M Wilde, Nipu Datta, Min-Hsiu Hsieh, and Andreas Winter.
\newblock Quantum rate-distortion coding with auxiliary resources.
\newblock {\em IEEE Transactions on Information Theory}, 59(10):6755--6773,
  2013.

\bibitem{Winter2002}
Andreas Winter.
\newblock Compression of sources of probability distributions and density
  operators.
\newblock {\em arXiv preprint quant-ph/0208131}, 2002.

\end{thebibliography}
}

\appendices

\section{Lemmas for the proof of the Bottleneck rate function \label{average}}
The following lemma is relevant for the proof of achievability of the communication rate. It has the same application as in Theorem 19 of \cite{Winter2002} and Lemma 1 of \cite{Datta2013}, but has been adapted to the distortion criterion for the quantum Bottleneck method.

\begin{lemma}\label{lemmaB}
\label{03apprevshan}
Let $\eta(\lambda) := - \lambda \log \lambda$. There exists a constant $k >0$ depending only on the dimension of $\mathcal{H}_{\tilde X Y}$ such that the following holds:

Let $0 < J \leq I(X;Y)_\rho$ be fixed. Let a quantum channel $\mathcal{N}_{X \to \tilde X}$ be such that if we apply the channel to the system $X$ and an identity channel $\mathcal{I}_{Y}$ on the system $Y$
the effect will meet the condition $I(\tilde{X};Y)_{\tilde\rho}\geq J$, where $\tilde\rho_{\tilde X Y}$ is given by (\ref{rho}).
Further, let $\mathcal{F}_n$ be a sequence of quantum channels from the space of density matrices
$\mathfrak{D}(\mathcal{H}_{X'X}^{\otimes n})$ to $\mathfrak{D}(\mathcal{H}^{\otimes n}_{X' \tilde X})$ such that 
\begin{equation}\label{tausigma}
\| \tilde\sigma_{X'^n \tilde X^n} - \tilde\tau^{\otimes n}_{X'\tilde X} \|_1 \leq \epsilon
\end{equation}
for some $0\leq \epsilon < \frac{1}{e}$ and large enough $n$, where $\tilde\tau_{X' \tilde X}$ and $\tilde\sigma_{X'^n \tilde X^n}$ are given by (\ref{tau}) and (\ref{sigma}), respectively.

Then, for large enough $n$ and $\delta:=k  \eta(\epsilon)$, we have 
\EQ{
 \bar{I}_n(\tilde X; Y)_{\tilde\sigma} \geq J - \delta. \label{03appfst}
}

\end{lemma} 

\begin{proof}

Adding and subtracting $I(\tilde X; Y)_{\tilde\rho}$ to the left hand side of Eq.~(\ref{03appfst}) and using the triangle inequality, we obtain
\EQ{
  \bar{I}_n(\tilde X; Y)_{\tilde\sigma}&=&\vert I(\tilde X; Y)_{\tilde\rho} -  \big(I(\tilde X; Y)_{\tilde\rho}-\bar{I}_n(\tilde X;  Y)_{\tilde\sigma}\big)  \vert \\
&\geq& \vert I(\tilde X; Y)_{\tilde\rho} \vert - \vert I(\tilde X ; Y)_{\tilde\rho}-\bar{I}_n(\tilde X; Y)_{\tilde\sigma}\vert.
}
Since we assumed that $I(\tilde{X};Y)_{\tilde\rho} \geq J$, what we need to show is that
\begin{equation}
\vert  I(\tilde X ;  Y)_{\tilde\rho} -\bar{I}_n(\tilde X ;  Y)_{\tilde\sigma}\vert \leq k  \eta(\epsilon) =:\delta 
\end{equation}
with some $k>0$ and $\eta$ as above. Inserting $\bar{I}_n$ given in Eq.~\eqref{barI}, we have
\EQ{
\nonumber \vert  I(\tilde X ;Y)_{\tilde\rho}-\bar{I}_n(\tilde X ; Y)_{\tilde\sigma} \vert =&\Big\vert  \frac{1}{n}  n \Big( S(\tilde\rho_{\tilde X}) + S(\rho_{ Y}) \\
&- S(\tilde\rho_{\tilde X  Y})\Big) -\frac{1}{n} \sum_{i=1}^n \Big( S(\tilde\sigma_{\tilde X_i}) \\
&+ S(\tilde\sigma_{ Y_i}) -S(\tilde\sigma_{\tilde X_i  Y_i}) \Big)  \Big\vert.
}
Hence, it suffices to show that for all $1\leq i \leq n$,
\EQ{
\nonumber \vert S(\tilde\rho_{\tilde X})&+ S(\rho_{ Y}) - S(\tilde\rho_{\tilde X Y})-S(\tilde\sigma_{\tilde X_i})\\
&-S(\tilde\sigma_{ Y_i})+S(\tilde\sigma_{\tilde X_i  Y_i})   \vert \leq k  \eta(\epsilon). \label{suff}
}
In particular, we will prove bounds of the above type on
$\vert S(\tilde\rho_{\tilde X})-S(\tilde\sigma_{\tilde X_i})  \lvert$ , 
$ \vert S(\tilde\rho_{\tilde X  Y})-S(\tilde\sigma_{\tilde X_i  Y_i}) \lvert$
and $\vert S(\rho_{ Y})-S(\tilde\sigma_{Y_i}) \vert $ for all $i$.
We start with the first of these. 
To prove this inequality, we recall that  $\tau_{X' X}$ is a purification of $\rho_x$.
Now let $\rho_{X' Y' XY}$ be a purification of $\rho_{XY}$, then there is a Hilbert space $ \mathcal{\hat H}$ and a unitary 
$U : \mathcal{H}_{X'} \oplus \mathcal{\hat H} \rightarrow \mathcal{H}_{X'Y'Y}$ such that
\begin{equation}
(I_X \otimes U)\tau_{X'X} (I_X \otimes U)^{\dagger} =\rho_{X'Y'XY},
\end{equation}
where $\tau_{X' X}$ is extended to the orthogonal complement of $\mathcal{H}_{X'}$ by zeros.
Then,  (\ref{tausigma}) implies
\begin{equation}
\| \tilde\rho^{\otimes n}_{X' Y' \tilde X  Y }- (I_X \otimes U^{\dagger})^{\otimes n } \tilde\sigma_{X'^n \tilde X^n}(I_X \otimes U)^{\otimes n} \|_1 \leq \epsilon,
\end{equation}
where $\tilde\rho_{X' Y' \tilde X  Y } = (\mathcal{N}_{X \rightarrow \tilde X} \otimes \mathcal{I}_{X'Y'Y})(\rho_{X'Y'X Y})$. Further, one computes that
\begin{equation}
(I_X \otimes U^{\dagger})^{\otimes n } \tilde\sigma_{X'^n \tilde X^n}(I_X \otimes U)^{\otimes n} = \tilde\sigma_{X'^n Y'^n \tilde X^n Y^n},
\end{equation}
\begin{equation}
\tilde\sigma_{X'^n Y'^n \tilde X^n  Y^n} := (\mathcal{I}_{X'^n Y'^n Y^n} \otimes \mathcal{F}_n)(\rho^{\otimes n}_{X'Y'XY}).
\end{equation}
To summarize, we found that
\begin{equation}
\| \tilde\rho^{\otimes n}_{X' Y' \tilde X  Y }- \tilde\sigma_{X'^n Y'^n \tilde X^n  Y^n} \|_1 \leq \epsilon.
\end{equation}
Using monotonicity of the trace norm under partial trace, we find that
\begin{equation}
 \|  \tilde\rho_{X' Y' \tilde X  Y}-\tilde\sigma_{X'_i Y'_i \tilde X_i Y_i} \|_1 \leq \|  \tilde\rho^{\otimes n}_{X' Y' \tilde X Y}-\tilde\sigma_{X'^n Y'^n \tilde X^n   Y^n} \|_1.
\end{equation}
Moreover, 
\begin{equation}
\| \tilde\rho_{ \tilde X } -\tilde\sigma_{ \tilde X_i }\|_1 \leq \| \tilde\rho_{X' Y' \tilde X Y} -\tilde\sigma_{X'_i Y'_i \tilde X_i  Y_i} \|_1. 
\end{equation}
again using the monotonicity of the trace norm under partial trace.
This implies that $\| \tilde\rho_{\tilde X} -\tilde\sigma_{\tilde X_i}\|_1 \leq \epsilon$.

Now, by the Fannes Inequality the following bound holds for all $1\leq i \leq n$:
\EQ{
\nonumber \vert S(\tilde\rho_{\tilde X}) -S(\tilde\sigma_{\tilde X_i}) \vert \leq &  \log(k')  \|  \tilde\rho_{\tilde X}-\tilde\sigma_{\tilde X_i} \|_1 \\
&+ \frac{1}{\log(2)}  \eta \big( \| \tilde\rho_{\tilde X} -\tilde\sigma_{\tilde X_i}\|_1 \big), 
}
where $k'$ is the dimension of $\mathcal{H}_{\tilde X}$. Then, using the bound $\| \tilde\rho_{\tilde X}-\tilde\sigma_{\tilde X_i}  \|_1 \leq \epsilon$, we find for all $1\leq i \leq n$,
\begin{align}
\vert S(\tilde\rho_{\tilde X})-S(\tilde\sigma_{\tilde X_i}) \vert &\leq  \epsilon\log(k')  + \frac{\eta(\epsilon)}{\log(2)} \nonumber\\
& \leq \hat k'  \eta(\epsilon), \label{A1}
\end{align}
where the last inequality uses the fact that  $\eta(\epsilon) \geq \epsilon$ for $0\leq \epsilon < \frac{1}{e}$, and where $\hat k'$ is defined in terms of the constants in the first inequality, including $k'$.

With a similar method, one can prove that for all $1\leq i \leq n$,
\begin{align}
\vert  S(\tilde\rho_{\tilde X  Y}) -S(\tilde\sigma_{\tilde X_i  Y_i})\vert & \leq  \hat k''  \eta(\epsilon),\label{B1}\\
\vert  S(\rho_{ Y}) -S(\tilde\sigma_{ Y_i})\vert & \leq  \hat k'''  \eta(\epsilon),\label{C1}
\end{align}
where $\hat k'', \hat k'''$ also depend on the dimensions of $\mathcal{H}_{\tilde X  Y}$ and $\mathcal{H}_{ Y}$, respectively.

Combining Eq.~(\ref{A1}), (\ref{B1}) and (\ref{C1}) we find for all $1\leq i \leq n$ that Eq.~(\ref{suff}) holds for a constant $k$ that 
includes the constants in the three estimates above and depends only on the dimension of $\mathcal{H}_{\tilde X  Y}$.

Hence, we obtain for large enough $n$ that $\bar{I}_n(\tilde X;  Y)_{\tilde\sigma} \geq J - k  \eta(\epsilon)$.
\end{proof}

\section{Numerical examples}\label{numerics}

The aim of this appendix is to compute the communication rate as a function of $J$ for some examples, using a numerical optimisation algorithm for evaluating the right-hand side of Eq.~\eqref{qrate}. In particular,  in all these examples the rate function turns out to be convex in $J$. Consider the following normalized version of the rate 
\begin{equation}\label{Mn}
\hat R(J) := \min_{\mathcal{N}_{X \to \tilde X}: \frac{I(\tilde X;Y)_{\tilde\rho}}{I(X;Y)_{\rho}}  \geq J} \frac{I(X';\tilde X)_{\tilde\tau}}{I(X';X)_{\tau}},
\end{equation}
where now $0< J < 1$. In the following examples the systems $X$ and $Y$ are described by two-dimensional Hilbert spaces spanned by the basis $\lvert \uparrow \rangle, \lvert \downarrow \rangle$.

To find the optimum numerically, a simple random search algorithm is used \cite{code}. It initially chooses a number of channels at random (in terms of their Kraus operators) and computes the related mutual information, then randomly varies those channels with the lowest $I(X',\tilde X)$ further until a stable optimum is reached. 

This algorithm is applied to three classes of input states $\rho_{XY}$: The first example is a ``classical'' state, i.e., a state without entanglement between the systems $X$ and $Y$, given by the density matrix 
\begin{equation*}
\rho_{XY}^{(1)} := p_1 \lvert \uparrow \uparrow \rangle \langle \uparrow \uparrow \rvert + p_2 \lvert \downarrow \uparrow \rangle \langle \downarrow \uparrow \rvert
+ p_3 \lvert \uparrow \downarrow \rangle \langle \uparrow \downarrow \rvert + p_4 \lvert \downarrow \downarrow \rangle \langle \downarrow \downarrow \rvert~,
\end{equation*}
where $p_1, p_2, p_3, p_4$ are nonnegative numbers with $p_1 + p_2 + p_3 + p_4 = 1$. 
The second example is a state with entanglement between $X$ and $Y$, namely, 
\begin{equation*}
\rho_{XY}^{(2)} = \frac{1}{2} \lvert \uparrow \uparrow \rangle \langle \uparrow \uparrow \rvert +  \frac{1}{4} \lvert \uparrow \uparrow \rangle \langle \downarrow \downarrow \rvert +
 \frac{1}{4} \lvert \downarrow \downarrow \rangle \langle \uparrow \uparrow \rvert +  \frac{1}{2} \lvert \downarrow \downarrow \rangle \langle \downarrow \downarrow \rvert.
\end{equation*}
Finally, the third example is again a state with entanglement defined as
\begin{equation*}
\rho_{XY}^{(3)} = p_1 \lvert v \rangle \langle v \rvert + p_2 \lvert w \rangle \langle w \rvert
\end{equation*}
with the normalized vectors
\begin{equation*}
v = \frac{1}{\sqrt{2}}\big( \lvert\uparrow \uparrow \rangle + \lvert \downarrow \downarrow \rangle \big), \quad w = \lvert \downarrow \downarrow \rangle
\end{equation*}
and nonnegative numbers $p_1, p_2$ with $p_1 + p_2 =1$. 
The  plots presented in Figs.~\ref{state1}--\ref{state3bis} show the rate as a function of distortion. The blue lines correspond to the curves $\hat R(J) = J$ and $\hat R(J) = \frac{1}{2}J$ and are introduced for comparison with the actual result (red line).

\begin{figure}[t!]
  \centering
      \includegraphics[width=0.45\textwidth]{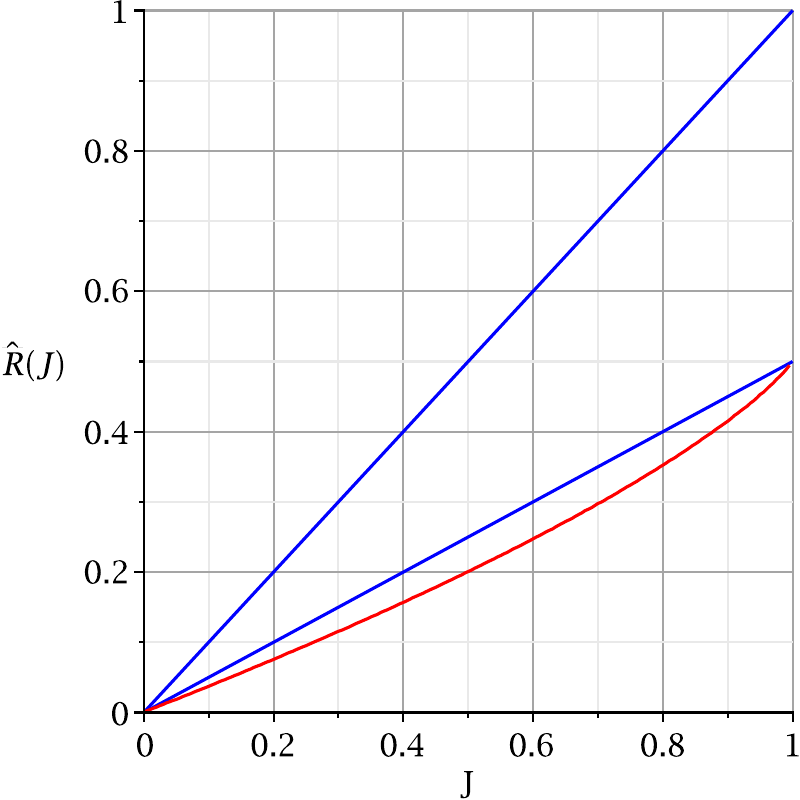}
  \caption{\small The function $\hat R(J)$ (in red) for the initial state $\rho_{XY}^{(1)}$ with $(p_1, p_2, p_3, p_4) = (0.1,0.2,0.3,0.4)$.}
  \label{state1}
\end{figure}

\begin{figure}[t!]
  \centering
      \includegraphics[width=0.45\textwidth]{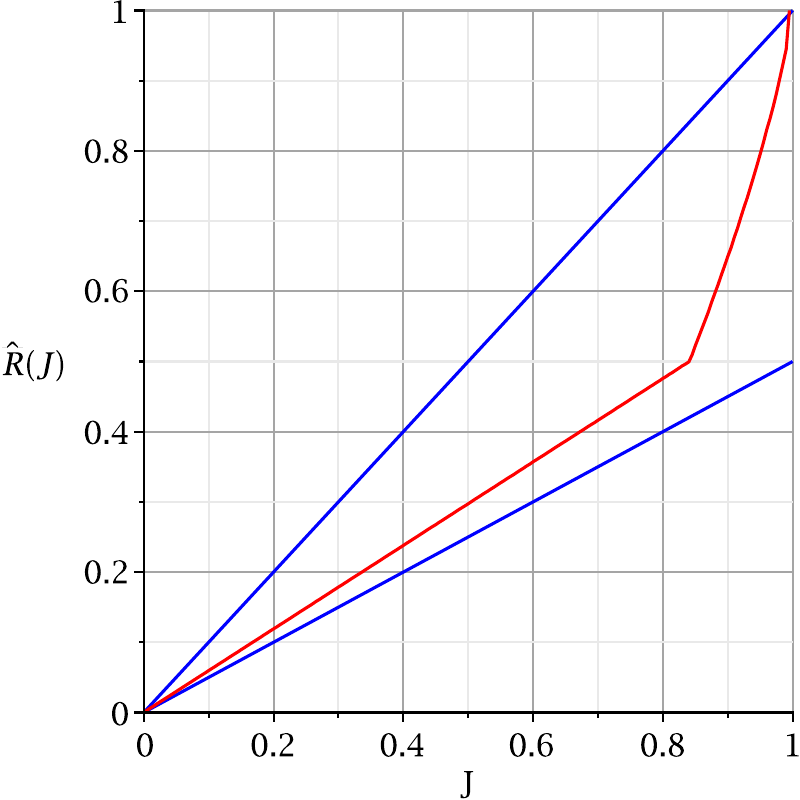}
  \caption{\small The function $\hat R(J)$ (in red) for the initial state $\rho_{XY}^{(2)}$.}
  \label{state2}
\end{figure}

\begin{figure}[t!]
  \centering
      \includegraphics[width=0.45\textwidth]{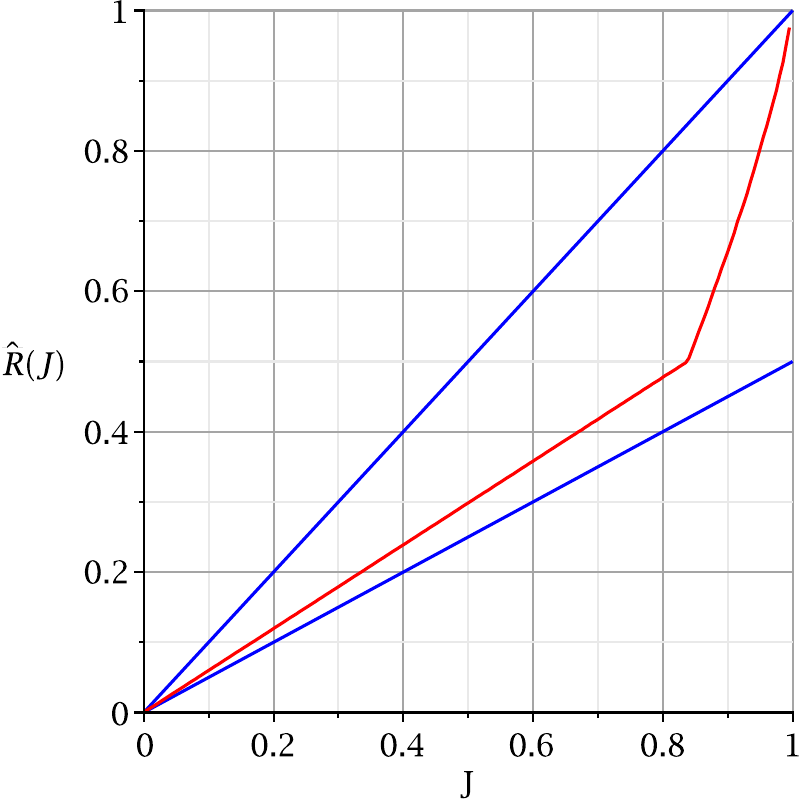}
  \caption{\small The function $\hat R(J)$ (in red) for the initial state $\rho_{XY}^{(3)}$ with $(p_1,p_2)= (0.4,0.6)$.}
  \label{state3}
\end{figure}

\begin{figure}[t!]
  \centering
      \includegraphics[width=0.45\textwidth]{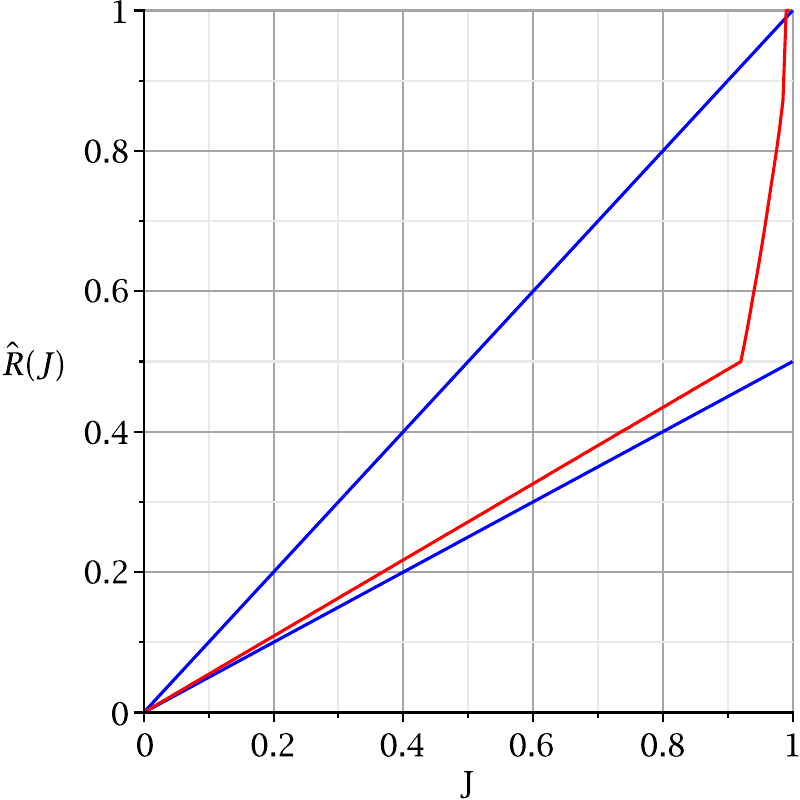}
  \caption{\small The function $\hat R(J)$ (in red) for the initial state $\rho_{XY}^{(3)}$ with $(p_1,p_2)= (0.2,0.8)$.}
  \label{state3bis}
\end{figure}

The four plots show that the function $\hat R(J)$, given by \eqref{Mn}, is indeed a convex function in $J$ for the specific choices of initial state $\rho_{XY}$  and within the limits of the numerical approximation. Note that in Figs.~\ref{state2}--\ref{state3bis},  the graph does not appear to be differentiable at the point $\hat R=\frac{1}{2}$. This seems to be a common feature of the examples $\rho_{XY}^{(2)}, \rho_{XY}^{(3)}$, but we do not currently have an analytic explanation for this behaviour. Note that in Fig.~\ref{state1}, one has $M(1) = \frac{1}{2}$, while in Figs.~\ref{state2}--\ref{state3bis}, one has $M(1) = 1$. In other words, in the case of a ``classical'' (non entangled) state $\rho_{XY}$, there is a channel such that $I(\tilde X; Y) = I(X;Y)$ and $I(\tilde X; X') = \frac{1}{2} I(X;X')$. To obtain an analytic expression of this channel we proceed as follows.

The initial state is $\rho_{XY}^{(1)}$.
Then, $\rho_X^{(1)} = \operatorname{Tr}_{Y}\rho_{XY}^{(1)} = (p_1+p_3)\lvert \uparrow \rangle \langle \uparrow \rvert + (p_2+p_4)\lvert \downarrow \rangle \langle  \downarrow \rvert$ and a purification is given by  $\tau_{X'X}^{(1)} = \lvert w \rangle \langle w \rvert$ with $\lvert w \rangle = \sqrt{p_1+p_3}\lvert \uparrow \uparrow \rangle + \sqrt{p_2+p_4}\lvert \downarrow \downarrow \rangle$.
Our ansatz for the channel is 
\begin{equation}\label{ch}
\mathcal{N}_{X \rightarrow \tilde X}(\rho) = \frac{1}{2} \sum_{i=1}^2 K_i \rho K_{i}^{\dagger}
\end{equation}
with $K_1 := \lvert \uparrow \rangle \langle \uparrow \rvert - \lvert \downarrow \rangle \langle \downarrow \rvert$ and 
 $K_2 :=\lvert \uparrow \rangle \langle \uparrow \rvert + \lvert \downarrow \rangle \langle \downarrow \rvert$. Applying this channel to $\rho_{XY}^{(1)}$ and $\tau_{X'X}^{(1)}$, we obtain $\tilde\rho_{\tilde X Y}^{(1)} = \rho_{XY}^{(1)}$ and
\begin{equation*}
\tilde\tau_{X'\tilde X}^{(1)} =  (p_1 + p_3)\lvert \uparrow \uparrow \rangle \langle \uparrow \uparrow \rvert_{X' \tilde X }
+  (p_2+p_4) \lvert \downarrow \downarrow \rangle \langle \downarrow \downarrow \rvert_{X'\tilde X }.
\end{equation*}
Note that here $\tilde\tau_{\tilde X}^{(1)} = \tau_{X'}^{(1)} = \tau_X^{(1)} = \rho_X^{(1)}$ are all matrices, but they act on different Hilbert spaces. 
With this choice of channel and initial state it is clear  that $I(\tilde X ; Y)_{\tilde\rho^{(1)}} = I(X;Y)_{\rho^{(1)}}$. In order to show that $I(\tilde X; X')_{\tilde\tau^{(1)}} = \frac{1}{2}I(X;X')_{\tau^{(1)}}$ we compute the von Neumann entropies in $I(\tilde X;X') = S(\tilde X) + S(X') - S(\tilde X X')$ and in $I(X;X')= S(X) + S(X') - S(XX')$ using the fact that $S(XX')=0$, since $\tau_{X'X}$ is pure, and that $\rho_X^{(1)}, \tilde\tau_{X'\tilde X}^{(1)}$ are diagonal. After evaluating the matrix functions, we obtain
\EQ{
 I(\tilde X; X') =&-(p_1 + p_3) \log(p_1 + p_3)\\
 &- (p_2 + p_4)\log(p_2 + p_4)\\
=& \frac{1}{2} I(X;X')
}
as desired. Therefore, $\hat R(1) \leq 1/2$ for this class of states, consistent with the graph. 

This somewhat unexpected feature may be understood as follows: In order to transmit the $X$ part of the non-entangled state $\rho_{XY}^{(1)}$ perfectly, a \emph{classical} channel of 1 bit capacity is sufficient. By the usual quantum teleportation result, this  corresponds to a quantum channel of only $\frac{1}{2}$ qbit capacity, if shared entanglement is available in abundance.

{While the simplistic random search algorithm that we used in this appendix is sufficient for demonstrating the essential quantum features of the system, it would be too inefficient to compute the rate function $R(J)$ for systems with larger Hilbert space dimension. More sophisticated numerical schemes would be required to that end, e.g., taking the gradient of the objective function into account. Appendix~\ref{BA} describes a Lagrangian method that may be used to this end.}

\section{The Optimal Map \label{BA}}

In this appendix we describe a method for finding the optimum channel in Eq.~\eqref{qrate}. We first remark that the optimum is actually attained at the boundary of the constraint region, $I(\tilde X;Y) = J$. This is heuristically clear, as for a channel in the interior ($I(\tilde X;Y) > J$) one could always reduce the channel's communication rate by a small amount at the expense of a slightly decreased quality measure $I(\tilde X;Y)$, staying inside the region of the constraint. We give a more formal argument to this end.
 
{\begin{lemma}\label{lemma:boundarY}
The minimum in Eq.~\eqref{qrate} is attained on the hypersurface $I(\tilde X;Y) = J$.
\end{lemma}
\begin{proof}
Suppose that the minimum was attained at a channel $\mathcal{N}_1$ such that $I(\tilde X;Y)_{\mathcal{N}_1} > J$ (here the subscript denotes that the state $\tilde\rho_{\tilde X Y}$ is computed with respect to the channel $\mathcal{N}_1$).  
Let $\mathcal{N}_0$ be the channel $\mathcal{N}_0 (\sigma) := (\operatorname{dim} \mathcal{H}_{\tilde X})^{-1} \operatorname{Tr} ( \sigma ) \cdot \pmb{1}$; we note that $I(\tilde X; Y)_{\mathcal{N}_0} =0=I(\tilde X; X')_{\mathcal{N}_0}$. For $0 < \lambda < 1$, let us consider a new channel $\mathcal{N}_\lambda:= \lambda \mathcal{N}_1 + (1 - \lambda)\mathcal{N}_0$. From convexity of the mutual information in the channel, it follows that
\begin{equation}
\begin{aligned}
I(\tilde X; X')_{\mathcal{N}_\lambda} &\leq \lambda I(\tilde X;X')_{\mathcal{N}_1} + (1-\lambda) I(\tilde X; X')_{\mathcal{N}_0} 
\\ &<  I(\tilde X;X')_{\mathcal{N}_1}.
\end{aligned}
\end{equation}
On the other hand, from continuity considerations we know that  $I(\tilde X;Y)_{\mathcal{N}_\lambda} > J$ for $\lambda$ close enough to 1, so that also $\mathcal{N_\lambda}$ fulfills the constraint in this region.  Hence $M(J) <  I(\tilde X;X')_{\mathcal{N}_1}$, and $\mathcal{N}_1$ is not the position of the minimum.
\end{proof}
}

{Therefore, we deal with an optimisation problem of a function over a constraint hypersurface, which makes it useful to look for local extrema with the method of Lagrange multipliers.
We use the Lagrangian in Eq.~(\ref{lag}),
\EQ{
\nonumber \mathcal{L}:=& I(X';\tilde X)_{\tilde\tau_{X'\tilde X}} \\
&-\beta I(\tilde X; Y)_{\tilde\rho_{\tilde X Y}}-\operatorname{Tr}_{X \tilde X}(\Psi_{X \tilde X}^{t_{X}}(\Lambda_{X} \otimes I_{\tilde X} )).
}
and look for zeros of its derivatives.}
We use the Choi-Jamio\l{}kowski representation
\EQ{
\Psi_{X' \tilde X}:=\big( \mathcal{I}_{X'} \otimes \mathcal{N}_{X \to {\tilde X} }\big) (\Phi_{X'X})
}
of the channel in order to compute the derivative of the Lagrangian, where $\Phi_{X'X}:=\sum_{i,j=0}^{d-1}\ketbra{i}{j}_{X'} \otimes \ketbra{i}{j}_{X}$ is the Choi-Jamio\l{}kowski matrix corresponding to the identity channel from the Hilbert space $\mathcal{H}_{X'}$ to $\mathcal{H}_X$,
and $\mathcal{N}_{X \to {\tilde X}}$ is the channel that simulates the compression-decompression process. $\Lambda_{X}$, an operator on the Hilbert space $\mathcal{H}_{X}$,
 is the Lagrange multiplier introduced for the normalisation of the channel $\Psi_{X \tilde X}^{t_{X}}$. Considering the definition of the mutual information, to 
compute the derivative $\frac{\delta \mathcal{L}}{\delta \Psi_{X \tilde X}^{t_{X}}}$ of the Lagrangian, we need to compute the following derivatives,
\EQ{
\frac{\delta S(X')_{\tilde\tau}}{\delta \Psi_{X \tilde X}^{t_{X}}}, \quad  \frac{\delta S(Y)_{\tilde\rho}}{\delta \Psi_{X \tilde X}^{t_{X}}}, \label{zeros}\\
 \frac{\delta S({\tilde X})_{\tilde\tau}}{\delta \Psi_{X \tilde X}^{t_{X}}}, \quad  \frac{\delta S({\tilde X})_{\tilde\rho}}{\delta \Psi_{X \tilde X}^{t_{X}}}, \quad \frac{\delta S(X'{\tilde X})_{\tilde\tau}}{\delta \Psi_{X \tilde X}^{t_{X}}}, \\
\frac{\delta S({\tilde X}Y)_{\tilde\rho}}{\delta \Psi_{X \tilde X}^{t_{X}}}, \quad  \frac{\delta \operatorname{Tr}_{X \tilde X}(\Psi_{X \tilde X}(\Lambda_{X} \otimes I_{\tilde X} ))}{\delta \Psi_{X \tilde X}^{t_{X}}}.
}

Notice that the functions in the numerator of the expressions in the first two equations in (\ref{zeros}) are independent of the channel and, hence, the derivatives are zero. For the five remaining ones we note that for a Hermitian operator, $A$, and a function, $f$, which is analytic on the spectrum of $A$, the
directional derivative of $\operatorname{Tr}[f(A)]$ is given by 
\EQ{
\frac{\delta \operatorname{Tr}[f(A)]}{\delta A}[B]=\operatorname{Tr}[f'(A)B], \label{variation}
}
with the direction given by the operator $B$ and $f'$ being the first derivative of the function $f$. (This follows from analytic functional calculus, expanding $\operatorname{Tr} f(A+\epsilon B) $ in a Taylor series around $\epsilon=0$.) Specifically, let us define $f(z):=z\log (z)$. Since the derivative of our function $f$ is given by $f'(\cdot)=(1+\log)(\cdot)$, using (\ref{variation}) and (\ref{03choijoint}), we have
\begin{equation}
\dfrac{\delta S(\tilde X)_{\tilde\tau}}{\delta \Psi_{X \tilde X}^{t_{X}}} [B_{X \tilde X}] = -\operatorname{Tr}_{X' \tilde X} \big\{ \big[ (I_{\tilde X} + \log \tilde\tau_{\tilde X}) \otimes I_{X'} \big] \mathbb{E} \big\},
\end{equation}
where
\begin{equation}
\mathbb{E} := \operatorname{Tr}_X \big(  B_{X \tilde X} \tau_{X' X} \big).
\end{equation}
Likewise, we can compute
\begin{equation}
\dfrac{\delta S(X'\tilde X)_{\tilde\tau}}{\delta \Psi_{X \tilde X}^{t_{X}}} [B_{X \tilde X}]=-\operatorname{Tr}_{X'\tilde X  } \Big\{  \big[(I_{X' \tilde X} + \log \tilde\tau_{X'\tilde X }) \big]\mathbb{E}\Big\}.
\end{equation}
In order to compute the derivative for  $S(\tilde X)_{\tilde\rho}$ and $S(\tilde X Y)_{\tilde\rho}$ we use Eqs.~(\ref{rhotildeXY}) and (\ref{variation}), and we find
\begin{equation}
\dfrac{\delta S({\tilde X})_{\tilde\rho}}{\delta \Psi_{x\tilde X}^{t_{X}}} [B_{X \tilde X}]=-\operatorname{Tr}_{\tilde X Y} \Big\{  \big[ \big( I_{\tilde X } + \log \tilde\rho_{{\tilde X}} \big) \otimes I_{Y} \big] \mathbb{G}  \Big\}
\end{equation}
and
\begin{equation}
\dfrac{\delta S({\tilde X}Y)_{\tilde\rho}}{\delta \Psi_{X \tilde X}^{t_{X}}} [B_{X \tilde X}]=-\operatorname{Tr}_{\tilde X Y} \Big\{  \big( I_{\tilde X Y} + \log \tilde\rho_{{\tilde X}Y} \big) \mathbb{G}  \Big\},
\end{equation}
where 
\EQ{
\mathbb{G}&:=&\operatorname{Tr}_{X}\big( B_{X \tilde X} \rho_{XY}  \big).
}
For the last derivative we have 
\EQ{
&\nonumber \dfrac{\delta \operatorname{Tr}_{X\tilde X}\big\{  \Psi_{X \tilde X}^{t_{X}} (\Lambda_ {X} \otimes I_{\tilde X})\big\}}{\delta \Psi_{X \tilde X}^{t_{X}}}[B_{X \tilde X}] &=\\
&\operatorname{Tr}_{X \tilde X}\big\{  (\Lambda_ {X} \otimes I_{\tilde X})B_{X \tilde X}\big\}. \label{normal}
}
Putting all the terms together we have
\EQ{
 \frac{\delta \mathcal{L}}{\delta \Psi_{X {\tilde X}}^{t_{X}}}[B_{X \tilde X}]=& \frac{\delta S({\tilde X})_{\tilde\tau}}{\delta \Psi_{X \tilde X}^{t_{X}}}[B_{X \tilde X}]-\frac{\delta S(X'{\tilde X})_{\tilde\tau}}{\delta \Psi_{X \tilde X}^{t_{X}}}[B_{X \tilde X}]\\&-\beta\frac{\delta S({\tilde X})_{\tilde\rho}}{\delta \Psi_{X \tilde X}^{t_{X}}}[B_{X \tilde X}] +\beta\frac{\delta S({\tilde X}{\tilde Y})_{\tilde\rho}}{\delta \Psi_{X \tilde X}^{t_{X}}}[B_{X \tilde X}]\\
 &- \frac{\delta \operatorname{Tr}_{X \tilde X}(\Psi_{X \tilde X}^{t_{X}}(\Lambda_ {X} \otimes I_{\tilde X}))}{\delta \Psi_{X \tilde X}^{t_{X}}}[B_{X \tilde X}]\\
 =& \operatorname{Tr}_{X'\tilde X } \Big\{  \big[ - I_{X'}\otimes (I_{\tilde X} + \log \tilde\tau_{\tilde X}) \\
 &+ (I_{X' \tilde X} + \log \tilde\tau )+\beta I_{X'} \\
 &\otimes (I_{\tilde X} + \log \tilde\tau_{\tilde X})\big] \mathbb{E}\Big\} \\
 &-\beta\operatorname{Tr}_{\tilde X Y} \Big\{ (I_{\tilde X Y} + \log \tilde\rho_{{\tilde X}Y})\mathbb{G}  \Big\}\\
 &-\operatorname{Tr}_{X \tilde X}\big\{ (\Lambda_ {X} \otimes I_{\tilde X})B_{X \tilde X}\big\}.
}
Setting this expression to zero ($\frac{\delta \mathcal{L}}{\delta \Psi_{x{\tilde X}}^{t_X}}[B_{X \tilde X}]=0$), 
we find
\EQ{
 &\operatorname{Tr}_{X'\tilde X }\big\{ \log \tilde\tau_{X' \tilde X}  \mathbb{E}\big\} =\\
&\operatorname{Tr}_{X'\tilde X } \Big\{  \big[ I_{X'}\otimes  \log \tilde\tau_{\tilde X}- \beta I_{X'}\otimes \log \tilde\tau_{\tilde X} \big] \mathbb{E}\Big\} \\
&+\beta\operatorname{Tr}_{\tilde X Y} \Big\{  \log \tilde\rho_{{\tilde X}Y} \mathbb{G}  \Big\}+\operatorname{Tr}_{X \tilde X}\big\{  (\Lambda_ {X} \otimes I_{\tilde X})B_{X \tilde X}\big\}.
}
Rearranging left and right hand sides of this equation, we find
\EQ{
&\operatorname{Tr}_{X \tilde X}\Big\{  B_{X \tilde X} \operatorname{Tr}_{X'} \Big\{ \tau_{X' X} ( \log \tilde\tau_{X' \tilde X} \otimes I_{X}) \Big\} \Big\}=\\
&\operatorname{Tr}_{X \tilde X}\Big\{  B_{X \tilde X} \operatorname{Tr}_{X'} \Big\{ \tau_{X' X} \big[ \big( I_{X'} \otimes \log \tilde\tau_{\tilde X}  - \beta I_{X'} \otimes \log \tilde\tau_{\tilde X} \big)\\
&\otimes I_X \big] \Big\}\Big\} \nonumber+\beta \operatorname{Tr}_{X \tilde X} \Big\{ B_{X \tilde X} \operatorname{Tr}_{Y} \Big\{   \rho_{XY} \big( \log \tilde\rho_{\tilde X Y} \otimes I_X \big) \Big\} \Big\}\\
&+\operatorname{Tr}_{X \tilde X} \Big\{ B_{X \tilde X}(\Lambda_ {X} \otimes I_{\tilde X})\Big\}.
}
This holds for all directions $B_{X \tilde X}$, which implies 
\EQ{
&\operatorname{Tr}_{X'} \Big\{ \tau_{X' X} ( \log \tilde\tau_{X' \tilde X} \otimes I_{X}) \Big\} =\\
&\operatorname{Tr}_{X'} \Big\{\tau_{X' X} \big[ \big( I_{X'} \otimes \log \tilde\tau_{\tilde X}\\
&- \beta I_{X'} \otimes \log \tilde\tau_{\tilde X} \big) \otimes I_X \big] \Big\}\\
&+\beta \operatorname{Tr}_{Y} \Big\{   \rho_{XY} \big( \log \tilde\rho_{\tilde X Y} \otimes I_X \big) \Big\} \\
&+\Lambda_ {X} \otimes I_{\tilde X}.
}
By performing the partial trace on the left hand side of this expression, we obtain
\EQ{
\rho_X^{1/2}( \log \tilde\tau_{X \tilde X})^{t_X}\rho_X^{1/2}=&\tau_{X}  \big( \log \tilde\tau_{\tilde X}  - \beta \log \tilde\tau_{\tilde X} \big)\\
&+\beta \operatorname{Tr}_{Y} \Big\{   \rho_{XY} \big( \log \tilde\rho_{\tilde X Y} \otimes I_X \big) \Big\}\\
&+\Lambda_ {X} \otimes I_{\tilde X}.\label{findlambda}
}
Simplifying this expression further, we find
\EQ{
(\log \tilde\tau_{X \tilde X})^{t_X} =& I_X \otimes \big( \log \tilde\tau_{\tilde X}  - \beta \log \tilde\tau_{\tilde X} \big)\\
&+\beta \operatorname{Tr}_{Y} \Big\{   \rho_X^{-1/2}\rho_{XY} \rho_X^{-1/2} \big( \log \tilde\rho_{\tilde X Y} \otimes I_X \big) \Big\}\\
&+\rho_X^{-1/2}\Lambda_ X \rho_X^{-1/2} \otimes I_{\tilde X}.\label{findtau}
}
Let us denote 
\EQ{
D_{X \tilde X}^{\beta Y} := &\beta I_X \otimes \log \tilde\tau_{\tilde X} \\
&- \beta\operatorname{Tr}_{Y} \Big\{   \rho_X^{-1/2}\rho_{XY} \rho_X^{-1/2} \big( \log \tilde\rho_{\tilde X Y} \otimes I_X \big) \Big\}, 
}
and the normalisation term $\tilde \Lambda_X := \rho_X^{-1/2}\Lambda_ X \rho_X^{-1/2}$. Exponentiating both sides of Eq.~\eqref{findtau}, we obtain
\EQ{
\tilde\tau_{X\tilde X }^{t_X}= e^{\log \tilde\tau_{\tilde X} \otimes I_X - D_{X \tilde X}^{\beta Y}+\tilde \Lambda_X \otimes I_{\tilde X} }.
}
Using Eq.~(\ref{03choijoint}), we arrive at the expression for the Choi-Jamio\l{}lkowski matrix corresponding to the channel, 
\EQ{\label{sol2}
\Psi_{X \tilde X}^{t_X} =&( \rho_X \otimes I_{\tilde X})^{-1/2} e^{\log \tilde\tau_{\tilde X} \otimes I_X - D_{X \tilde X}^{\beta Y}+\tilde \Lambda_X \otimes I_{\tilde X} }\\
&( \rho_X \otimes I_{\tilde X})^{-1/2}.
}
{Note that this is an implicit equation in $\Psi_{X \tilde X}$ since it also appears on the right hand side of this expression. To find the optimum channel,  Eq.~\eqref{sol2} needs to be solved  iteratively for $\Psi_{X\tilde X}$. 
Note that the unknown Lagrange multiplier $\Lambda_{X}$, which is associated with the normalisation constraint, is still contained in this equation. An algorithm that recursively computes the channel might work as follows. Starting with a guess for the channel $\Psi_{X \tilde X}$ and normalising it, this guess is inserted into Eq.~\eqref{findlambda} to compute a self-consistent value for $\Lambda_{X}$. This would allow to compute $\tilde\rho_{\tilde X}$ and $\tilde\rho_{\tilde X Y}$ from the channel, and hence give an approximation for all quantities that enter the right hand side of Eq.~\eqref{sol}. Thus,  a new approximation for the left-hand side is obtained, i.e., the channel $\Psi_{X \tilde X}$. When this procedure is  repeated iteratively,  an optimal channel $\Psi_{X \tilde X}$ is obtained at a given value for $\beta$.  Repeating this procedure for different values of $\beta$ and optimising under the constraint $I(\tilde X;Y)_{\tilde\rho}\geq J$ yields the minimum in \eqref{qrate}.}

\begin{IEEEbiographynophoto}{Sina Salek}
After completing his PhD at the University of Bristol, Sina Salek became a postdoctoral research fellow at the University of Hong Kong, before joining the Department of Computer Science, the University of Oxford as a postdoctoral researcher. His research interest is in quantum information theory and quantum foundations.
\end{IEEEbiographynophoto}

\begin{IEEEbiographynophoto}{Daniela Cadamuro}
Daniela Cadamuro leads an Emmy Noether junior research group at the Institute for Theoretical Physics at the University of Leipzig. She received a PhD in theoretical physics from the Georg-August-Universität Göttingen in Germany in 2012. She held postdoctoral positions at the University of Bristol and at the Technische Universität München (TUM). Her research interests lie in mathematical aspects of quantum theory, in particular quantum field theory.
\end{IEEEbiographynophoto}

\begin{IEEEbiographynophoto}{Philipp Kammerlander}
Philipp Kammerlander received the B.Sc. and M.Sc. degrees in Physics from ETH Zurich in 2012 and 2013. He is currently pursuing a Ph.D. degree at the Institute for Theoretical Physics at ETH Zurich. His research interests include classical and quantum thermodynamics and quantum information theory.
\end{IEEEbiographynophoto}

\begin{IEEEbiographynophoto}{Karoline Wiesner}
Karoline Wiesner is Associate Professor in the School of Mathematics at the University of Bristol. She received a PhD in physics from Uppsala University in Sweden in 2004. She was a postdoctoral fellow at the Santa Fe Institute and the University of California, Davis. In 2007 she joined the faculty of the School of Mathematics at the University of Bristol. Her research interests lie in quantum and classical information theory and complex systems, ranging from physics to the social sciences.
\end{IEEEbiographynophoto}

\end{document}